\documentclass[journal]{IEEEtran}

\ifCLASSINFOpdf
\else
   \usepackage[dvips]{graphicx}
\fi
\usepackage{url}

\usepackage{setspace}

\usepackage{soul}
\usepackage{color, xcolor}

\usepackage{setspace}

\usepackage{graphicx}

\usepackage{subfigure} 

\usepackage{amssymb}

\usepackage{amsmath}

\newenvironment{proof}{{\indent \indent \it Proof:}}{\hfill $\blacksquare$\par}
\newenvironment{notation}{{\indent \it Notation:}}

\newtheorem{theorem}{Theorem}
\newtheorem{lemma}{Lemma}
\newtheorem{proposition}{Proposition}
\newtheorem{remark}{Remark}

\usepackage{amsfonts}

\usepackage{graphicx}

\usepackage{epstopdf}

\usepackage{float}

\usepackage{stfloats}

\usepackage{indentfirst}

\usepackage[explicit]{titlesec}

\usepackage[citecolor=red]{hyperref}

\begin{document}

\title{\huge{Secure Communication for Spatially Correlated Massive MIMO with Low-Resolution DACs}}
\author{Dan Yang, \IEEEmembership{Student Member, IEEE,}
Jindan Xu, \IEEEmembership{Member, IEEE,}
Wei Xu, \IEEEmembership{Senior Member, IEEE,}
\\Ning Wang, \IEEEmembership{Member, IEEE,}
Bin Sheng, \IEEEmembership{Member, IEEE,}
and A. Lee Swindlehurst, \IEEEmembership{Fellow, IEEE}

\thanks{D. Yang, J. Xu, W. Xu, and B. Sheng are with the National Mobile Communications Research Laboratory, Southeast University, Nanjing 210096, China (email: {dyang@seu.edu.cn; jdxu@seu.edu.cn; wxu@seu.edu.cn; sbdtt@seu.edu.cn}).}

\thanks{N. Wang is with the School of Information Engineering, Zhengzhou University, Zhengzhou 450001, China.}

\thanks{A. Lee Swindlehurst is with the Center for Pervasive Communications and Computing, University of California at Irvine, Irvine, CA 92697 USA (e-mail: swindle@uci.edu).}}
\maketitle
\begin{abstract}
In this paper, the performance of a secure massive multiple-input multiple-output (MIMO) system adopting low-resolution digital-to-analog converters (DACs) is analyzed over spatially correlated wireless channels. A tight lower bound for the achievable secrecy rate is derived with artificial noise (AN) transmitted in the null space of the user channels. Using the analytical results, the impact of spatial correlation on the secrecy rate is explicitly evaluated in the presence of low-resolution DACs. The analytical observations reveal that using low-resolution DACs can be beneficial to the secrecy performance compared with ideal DACs, when the channels are strongly correlated and optimal power allocation is not employed. 
\end{abstract}

\setlength{\parskip}{0\baselineskip} 
\section{Introduction}

\IEEEPARstart{P}{hysical} layer security (PLS) has become an emerging technology for securing wireless communication without relying upon traditional cryptographic mechanisms. Compared to conventional upper-layer cryptographic schemes, PLS has the advantages of low computational complexity and low resource consumption [1]. Massive multiple-input multiple-output (MIMO) systems provide another disruptive technology for fifth generation (5G) cellular communications, and have shown great potential in improving spectral and energy efficiency. The use of large-scale antenna arrays in massive MIMO provides a large excess of redundant spatial degrees of freedom (DoF), which can be exploited to achieve secure physical layer transmission. This idea has been attracting increasing research interest in the past few years [2], [3]. 

Massive MIMO transmission requires a very high power consumption if high-resolution digital-to-analog converters (DACs) are employed in the RF chains for each antenna. At the transmitter, power expenditure is dominated by power amplifiers (PAs), which are usually required to operate within a high linearity regime to avoid distortion. A practical solution to the above challenge is to use low-resolution DACs, which relaxes the requirement of linearity and allows the amplifiers to operate closer to saturation, thus increasing the efficiency of PAs [4], [5]. In [6], both finite-bit DACs at base station (BS) and finite-bit analog-to-digital converters (ADCs) at user side were analyzed in the massive MIMO downlink. The work was then extended in [7] by considering spatially correlated channels. Further in [8], a constant envelope precoding technique was devised for the multiuser MIMO with one-bit DACs.

The effect of hardware impairments (HWIs) on spectral efficiency of massive MIMO systems has been studied in [9]. Regarding the secrecy performance, the authors in [10] analyzed the effects of HWIs on secrecy rate, where ideal converters with infinite resolution were considered. Secure communication in a massive MIMO system with low-resolution DACs was investigated in [11], which revealed that low-resolution DACs can achieve superior secrecy rate under certain conditions, e.g., at low SNR or with a large power allocation factor.

Most of the existing works on low-resolution DACs transmissions have focused on the assumption of independent identically distributed (i.i.d.) channels for massive MIMO. However, in practice, the limited space between the BS antennas as well as the rich scattering propagation environment can result in spatial correlation. The impact of correlated Rayleigh fading channels on optimal multiuser loading was analyzed in [6] by applying asymptotic random matrix theory. How spatial correlation impacts secure massive MIMO communication with low-resolution DACs is still an open problem.

In this paper, we focus on secure transmission in the massive MIMO downlink when low-resolution DACs are employed. A tight lower bound for the ergodic secrecy rate is derived that explicitly characterizes the impact of channel correlation on the secrecy rate for typical correlated channels. An optimal power allocation strategy is proposed, which suggests that more power should be allocated to AN when strong channel correlation is present. It is revealed that using low-resolution DACs can improve the secrecy performance for a fixed power allocation factor under strong spatially correlated channels.  

\begin{notation}
$\bf {X}^{*}$, $\bf {X}^{\mit T}$, $\bf{X}^{\mit H}$ and $\rm tr(\bf{X})$ represent the conjugate, transpose, conjugate transpose and trace of matrix $\bf{X}$, respectively. $\mathbb{E}\{\cdot\}$ is the expectation operator. $\rm diag(\cdot)$ denotes a diagonal matrix that retains only the diagonal elements of the input matrix, and $\widetilde{\rm diag}(\cdot)$ represents a diagonal matrix with the input vector as its diagonal entries.
\end{notation}

\vspace{-0.3cm}
\section{System Model}
The secure massive MIMO system under investigation comprises one $N$-antenna BS, $K$ single-antenna legitimate users, and one $M$-antenna passive eavesdropper. The channel matrices are modeled based on the Kronecker channel model as shown in [12]. To make the problem more tractable, we consider the system with a common correlation matrix at the BS. Specifically, the channel between the BS and the users is modeled as $\bf{H}=\bf{D}\rm^{\frac{1}{2}}\widetilde{\bf{H}}\bf{R}\rm^{\frac{1}{2}}$, where the elements of $\widetilde{\bf{H}}\in\mathbb{C}^{K\times N}$ are i.i.d. Gaussian random variables with zero mean and unit variance, the diagonal matrix $\bf{D}\in\mathbb{C}^{\mit K\times K}$ characterizes the large-scale fading with its $k$th diagonal element given by $\beta_k$, and $\bf{R} \in\mathbb{C}^{\mit N\times N}$ is the transmit covariance matrix satisfying $\rm tr(\bf{R})=\mit N$. Similarly, the channel matrix between the BS and the eavesdropper is $\bf{H}\rm_e=\bf{D}^{\rm \frac{1}{2}}\rm_e\widetilde{\bf{H}}_e\bf{R}^{\rm \frac{1}{2}}$, where $\widetilde{\bf{H}}\rm_e\in\mathbb{C}\mit^{M\times N}$ contains i.i.d. Rayleigh fading channel coefficients following $\mathcal{CN}(0,1)$. The diagonal matrix $\bf{D}\rm_e$ represents the large-scale fading at the eavesdropper with identical diagonal entries $\beta \rm^e$.

The BS desires to transmit the symbols $\bf{s}=[\mit s\rm_1,\mit s\rm_2,...,\mit s\mit_K]\in\mathbb{C}^{\mit K\times\rm1}$ to the legitimate users with $\mathbb{E}\{\bf{s}\bf{s}\mit ^H\}=\bf{I}\mit_K$ using a linear precoding matrix $\bf{W}\in\mathbb{C}^{\mit N\times K}$. The eavesdropper's channel state information (CSI) is assumed unknown to the BS, and AN is injected to ensure confidential communication. The AN vector $\bf{t}\sim\mathcal{CN}(\bf0,\bf I\mit_{N-K})$  is precoded by an AN shaping matrix $\bf{V}\in\mathbb{C}^{\mit N\times (N-K)}$. Denote by $P$ the total transmit power. The power allocation factor $\xi\in(0,1]$ aims to strike a balance between the transmit signal and the AN. The unquantized downlink transmit signal vector $\bf{x}$ is then expressed as
{\setlength\abovedisplayskip{3.4pt} 
\setlength\belowdisplayskip{3.4pt}
\begin{equation}
\bf{x}=\sqrt{\mu}\bf{Ws}+\sqrt{\nu}\bf{Vt},
\end{equation}
where $\mu\triangleq\frac{\xi P}{K}$ and $\nu\triangleq\frac{(1-\xi) P}{N-K}$.

The precoded signal is transmitted after DAC quantization, which is denoted by $\mathcal{Q}(\bf{x})$. Establishing the non-linear quantization model of a finite-bit DAC is challenging. We follow a popular way of charactering the quantizer by a linear function applying the simple additive quantization noise model. The quantized signal vector can accordingly be decomposed as
\begin{equation}
\bf{z}=\mathcal{Q}(\bf{x})=\mit\sqrt{\rm1-\rho} \bf{x}+\bf{q},
\end{equation}
where the quantization noise $\bf{q}$ is assumed to be uncorrelated with the input signal $\bf{x}$, and
\begin{equation}
\bf C_q=\mathbb{E}\{\bf{q}\bf{q}\mit^H\}=\rho\mathbb{E}\big\{\rm diag(\bf{x}\bf{x}\mit^H)\big\}.
\end{equation}
The value of the distortion factor $\rho$ depends on the DAC resolution; for example, it can be chosen as in [5] for DAC resolutions of less than 5 bits, or as $\rho=\frac{\sqrt{3}\pi}{2}\cdot2^{-2b}$ for scenarios with higher precision, where $b$ represents the number of quantization bits. From (1) and (3), the covariance matrix of the quantization noise equals
\begin{equation}
\bf{C_q}=\rho\big[\mu\rm diag(\bf{W}\bf{W}\mit^H)+\nu\rm diag(\bf{V}\bf{V}\mit^H)\big].
\end{equation}

Given the CSI of the legitimate channels, the matrix $\bf{V}$ is designed to lie in the null space of the channel matrix $\bf{H}$, i.e., $\bf{H}\bf{V}=\bf{0}$, which (ideally) makes the AN “invisible” to the legitimate users [13]. Using (1) and (2), the signals received at the users and the eavesdropper are expressed as
\begin{equation}
\bf{y}=\sqrt{1-\rho}(\sqrt{\mu}\bf{HWs}+\sqrt{\nu}\bf{HVt})+\bf{Hq}+\bf{n}
\end{equation}
\begin{equation} 
\bf{y}\rm_e=\sqrt{1-\rho}(\sqrt{\mu}\bf{H}\rm_e\bf{Ws}+\sqrt{\nu}\bf{H}\rm_e\bf{Vt})+\bf{H}\rm_e\bf{q}+\bf{n}\rm_e,
\end{equation} 
where $\bf{n}\sim\mathcal{CN}(\bf{0},\mit \sigma_n\rm^2\bf{I}\mit_K)$ and $\bf{n}\rm_e\sim\mathcal{CN}(\bf{0},\mit \sigma\rm_e\rm^2\bf{I}\mit_M)$ respectively represent the additive noise terms at the users and at the eavesdropper.

\section{Achievable Ergodic Secrecy Rate Analysis}
In this section, we derive a tight lower bound for the ergodic secrecy rate of the secure multiuser massive MIMO downlink and analyze the impact of spatial correlation on the secrecy rate in the presence of low-resolution DACs.

\subsection{Lower Bound on the Achievable Ergodic Secrecy Rate}
We adopt linear matched filter (MF) precoding for data transmission, i.e., $\bf{W}=\bf{H}/{\parallel \bf{H}\parallel}$. The received signal at the $k$th user according to (5) is expressed as
\begin{equation}
y_k=\sqrt{1-\rho}\big(\sqrt{\mu}\bf{h}\mit_k^T\bf{Ws}+\sqrt{\nu}\bf{h}\mit_k^T\bf{Vt}\big)+\bf{h}\mit_k^T\bf{q}+\mit n_k.
\end{equation}
Then, under the assumption of Gaussian distributed interference, a lower bound on the  ergodic rate for the $k$th user can be calculated as
\begin{small}
\begin{equation}
R_k=\mathbb{E}\big\{\rm log_2(1+\mit\gamma_k)\big\},
\end{equation}
\end{small}
\begin{small}
\begin{equation}
\gamma_k=\frac{(1-\rho)\mu\big|\bf{h}\mit_k^T\bf{w}\mit_k\big|\rm^2}{\varrho+\bf{h}\mit_k^T\bf{C_q}\bf{h}\mit_k^*+\rm(1-\rho)\nu\bf{h}\mit_k^T\bf{V}\bf{V}\mit^H\bf{h}\mit_k^*+\sigma\mit_n\rm^2},
\end{equation}
\end{small}where {\small$\varrho=(1-\rho)\mu\sum\limits_{j\neq k}\big|\bf{h}\mit_k^T\bf{w}\mit_j\big|\rm^2$}, $\bf{h}\mit_k^T$ denotes the $k$th row of $\bf{H}$, and $\bf{w}\mit_k$ is the $k$th column of $\bf{W}$.
Note that the numerator of $\gamma_k$ is the power of the desired signal component for the $k$th user, and the denominator represents the power from inter-user interference, quantization noise from the low-resolution DACs, AN leakage, and thermal noise.

\begin{lemma}
A lower bound on the achievable rate (8) of user $k$ is given by
\begin{small}
\begin{equation}
\underline{R}_k=\rm log_2\bigg(1+\frac{(1-\rho)\beta^2\mit_k\gamma\rm_0\xi \mit N/\sum_{i\rm=1}\mit^{K}\beta_i}{\varrho'+\rho\mit\beta_k\gamma\rm_0+1}\bigg),
\end{equation} 
\end{small}where {\small$\varrho'=(1-\rho)\xi\gamma_0 \beta\mit_k\rm tr(\bf{R}\rm^2)\mit\sum\limits_{j\neq k}\beta_j/(N\sum_{i=\rm1}\mit^{K}\beta_i)$}, and {\small$\gamma_0=\frac{P}{\sigma_n^2}$} is the average SNR.
\end{lemma}

\begin{proof}
Please refer to Appendix A. 
\end{proof}

To guarantee secure communication in the worst case, we assume that the eavesdropper has perfect CSI of all legitimate users and can remove all the interference from the legitimate users [2], [3], [10], [11]. According to (6), the ergodic rate of the eavesdropper is expressed as
\begin{equation}
C=\mathbb{E}\rm\bigg\{log_2\big(1+(1-\mit\rho)\mu\bf{w}\mit_k^H\bf{H}\rm_e\mit^H\bf{X}\rm^{-1}\bf{H}\rm_e\bf{w}\mit_k\big)\bigg\},
\end{equation} 
where $\bf{X}$ is defined as
\begin{equation}
\bf{X}=\rm(1-\rho)\nu\bf{H}\rm_e\bf{V}\bf{V}\mit^H\bf{H}\rm_e\mit^H+\bf{H}\rm_e\bf{C_q}\bf{H}\rm_e\mit^H.
\end{equation}
Furthermore, we assume that $\sigma\rm_e\rm^2$ is negligibly small corresponding to the worst case, and consequently, $C$ is independent of the path-loss of the eavesdropper $\beta\rm^e$ [2], [3], [10], [11]. A tight upper bound for $C$ is derived in the following lemma.

\begin{lemma}
For $N\rightarrow\infty$, an upper bound on the eavesdropping rate is given by
\begin{small}
\begin{equation}
\overline{C}=\rm log_2\bigg(1+\mit\frac{\phi M\xi\kappa\beta_k/\sum_{i=\rm1}\mit^{K}\beta_i}{\phi\kappa\rm^2\big(\frac{\mit N}{\rm tr(\bf{R}\rm^2)}-\mit a\big)-\varpi}\bigg),
\end{equation}
\end{small}where {\small$a=\frac{M}{N}$, $b=\frac{K}{N}$, $\rho'=\frac{\rho}{1-\rho}$, $\phi=1-b$, $\kappa=1-\xi+\rho'$}, and {\small$\varpi=ab(1-\xi)^2$}. 
\end{lemma}

\begin{proof}
Please refer to Appendix B.
\end{proof}

Applying \emph{Lemma 1} and \emph{Lemma 2}, a lower bound on the ergodic secrecy rate of the $k$th user is obtained in \emph{Theorem 1}.
 
\begin{theorem}
For $N\rightarrow\infty$, the achievable ergodic secrecy rate for the $k$th user is lower bounded by
\begin{small}
\begin{equation}
\underline{R}_{\rm sec}\triangleq{[}\underline{R}_k-\overline{C}{]}^+,
\end{equation}
\end{small}where $[x]^{+}=\rm max\{0,\mit x\}$, and $\underline{R}_k$ and $\overline{C}$ are given in (10) and (13), respectively. 
\end{theorem}

If no spatial correlation is present, i.e., $\bf R = I$, then (14) reduces to
\begin{small}
\begin{equation}
\begin{aligned}
\begin{split}
&\underline{R}_{\rm sec}=\bigg[\rm log_2\bigg(1+\frac{(1-\mit \rho)\beta_k\rm^2\rm\gamma_0\mit\xi N/\sum_{i=\rm1}\mit^{K}\beta_i}{(1-\rho)\gamma_0\beta\mit_k\mit\xi\sum\limits_{j\neq k}\beta_j/\sum_{i=\rm1}\mit^{K}\beta_i+\rho\beta\mit_k\rm\gamma_0+1}\bigg)
\\&-\rm log_2\bigg(1+\mit\frac{\phi M\xi\kappa\beta_k/\sum_{i=\rm1}\mit^{K}\beta_i}{\phi\kappa\rm^2(1-\mit a)-\varpi}\bigg)\bigg]^+.
\end{split}
\end{aligned}
\end{equation}
\end{small}As expected, $\underline{R}_{\rm sec}$ increases with $N$ and $\gamma_0$.

\subsection{Optimal Power Allocation Strategy for AN}
Here we investigate the impact of the power allocation factor on the ergodic secrecy rate in (14) under spatially correlated channels. Assume $\mit ab\ll\rm1$ in (13), which is reasonable in massive MIMO equipped with a large number of antennas. The derivative of $\underline{R}_{\rm sec}$ w.r.t. $\xi$ is calculated as
\begin{small}
\begin{equation}
\begin{aligned}
\begin{split}
&\frac{\partial \underline{R}_{\rm sec}}{\partial \xi}=\frac{L_1L_2}{\rm{ln2}\mit(L\rm_3\xi+\mit L\rm_2)[\mit L\rm_2+\xi(\mit L\rm_1+\mit L\rm_3)]}
\\&-\frac{M \rm tr(\bf{R}\rm^2)\rm(1+\rho')\mit\beta_k}{\rm ln2\big[\sum_{\mit i\rm=1}\mit^{K}\mit \beta_i\big(N- \rm tr(\bf{R}\rm^2)\mit a\big)\kappa\rm^2+\mit M\xi\beta_k \rm tr(\bf{R}\rm^2)\mit\kappa\big]},
\end{split}
\end{aligned}
\end{equation}
\end{small}where {\small$L_1=(1-\rho)\beta_k\rm^2\gamma_0 \mit N/\sum_{i\rm=1}\mit^{K}\beta_i$}, {\small$L_2=\rho\beta_k\gamma_0+1$}, and {\small$L_3=(1-\rho)\gamma_0 \beta_k\rm tr(\bf{R}\rm^2)\mit\sum\limits_{j\neq k}\beta_j/(N\sum_{i=\rm1}\mit^{K}\beta_i)$}. Since {\small$\frac{\partial \underline{R}_{\rm sec}}{\partial \xi}>0$} for small $\xi$ and {\small$\frac{\partial \underline{R}_{\rm sec}}{\partial \xi}<0$} for large $\xi$, the optimal power allocation factor $\xi^{*}$ that achieves the highest secrecy rate is obtained by setting {\small$\frac{\partial \underline{R}_{\rm sec}}{\partial \xi}=0$}. A closed-form expression for $\xi^{*}$ can be founded as follows:
\begin{small}
\begin{equation}
\xi^{*}=\frac{-B-\sqrt{B^2-4AC}}{2A},
\end{equation}
\end{small}where the parameters $A$, $B$, and $C$ are given by
\begin{small}
\begin{equation}
A=L_1L_2G_2-L_1L_2G_3-G_1L_3(L_1+L_3),
\end{equation}
\begin{equation}
B=(1+\rho')L_1L_2(G_3-2G_2)-G_1L_2(L_1+2L_3),
\end{equation}
\begin{equation}
C = G_2L_1L_2(1+\rho')^2-G_1L_2^2,
\end{equation}
\end{small}and {\small$G_1=M \rm tr(\bf{R}\rm^2)(1+\mit \rho')\beta_k$}, {\small$G_2=\sum_{i=1}^{K}\beta_i\big(N- \rm tr(\bf{R}\rm^2)\mit a\big)$}, and {\small$G_3=M\beta_k \rm tr(\bf{R}\rm^2)$}.

Assuming $\beta_k=1, 1\le k\le K$, we can simplify the above expressions to evaluate the impact of spatial correlation on $\xi^{*}$ for different DAC resolutions. Comparing the value of $\xi^{*}$ for the special case of an i.i.d. channel, i.e.,$\rm tr(\bf{R}\rm^2)=\mit N$ with a fully correlated channel, i.e., $\rm tr(\bf{R}\rm^2)=\mit N\rm^2$ for a Hermitian Toeplitz correlation matrix, we can easily observe that $\xi^{*}$ decreases when $\rm tr(\bf{R}\rm^2)$ increases from $N$ to $N^2$. The relationship between $\xi^{*}$ and the design parameters, including the DAC resolution and channel correlation coefficient, is verified in Section IV through numerical results.

\subsection{Impact of Spatial Correlation}
We first analyze the impact of the antenna ratio $a$ under the correlated channel condition when AN is injected. In (14), $\underline{R}_{\rm sec}$ decreases with respect to $a$. Considering the special case of $\beta_k=\beta, 1\le k\le K$, and $\xi\rightarrow0$, by setting $R\rm_{sec}=0$, the maximum number of eavesdropper antennas that still allows for a positive secrecy rate can be obtained from the following proposition.
\begin{proposition}
If a positive secrecy rate can be achieved, then the maximum antenna ratio $a$ is obtained as
\begin{small}
\begin{equation}
a\rm_{sec}=\frac{(1-\mit b)N\gamma\rm_0}{tr(\bf{R}\rm^2)\big[\gamma_0 \rho\mit b(\rho-\beta-\rm 2)+\gamma_0(1+\beta \rho)+1-\mit b\big]}.
\end{equation}
\end{small}
\end{proposition}

\begin{remark}
By direct inspection of (21), the maximum number of eavesdropper antennas that can be tolerated for secure transmission decreases with $\rho$ and the spatial correlation level because Eve can wiretap more information under strongly correlated channels. For the special case of $\rho\rightarrow0$ and $\rm tr(\bf{R}\rm ^2)=\mit N\rm^2$, we have $a\rm_{sec}=\frac{(1-\mit b)\gamma\rm_0}{\mit N(\rm1-\mit b+\gamma\rm_0)}$, which indicates that $a\rm_{sec}$ is independent of the large-scale fading factor with infinite-resolution DACs.
\end{remark}

To extract clear insights, we further consider a representative exponential correlation model [14]
{\setlength\abovedisplayskip{3pt}
\setlength\belowdisplayskip{3pt}
\begin{equation}
\bf{R}\mit_{ij}=\mit\zeta^{|i-j|},
\end{equation}
where $\zeta$ denotes the correlation coefficient. The exponential model is widely adopted in literature and is applicable to analysis for a massive MIMO system with uniform planar array (UPA) scenarios [15].

\begin{proposition}
The secrecy rate gap for different DAC resolutions decreases with the correlation coefficient $\zeta$.
\end{proposition}

\begin{proof}
$\lim\limits\mit_{N\to\infty}\frac{\rm tr(\bf{R}\rm^2)}{\mit N}=\rm\frac{1+\zeta^2}{1-\zeta^2}$ exists under the exponential correlation model in (22). From (14), we have {\small$\frac{\partial \underline{R}_{\rm sec}}{\partial \rho}=\frac{\partial \underline{R}_k}{\partial \rho}-\frac{\partial \overline{C}}{\partial \rho}$}. The first term {\small$\frac{\partial \underline{R}_k}{\partial \rho}$} is given by
\begin{small}
\begin{equation}
\frac{\partial \underline{R}_k}{\partial \rho}=-\frac{\beta^2_k \gamma_0 \xi N(1+\beta_k \gamma_0)\sum_{i=1}^{K}\beta_i}{\rm ln2(\Upsilon \mit\beta_k\gamma\rm_0+\sum_{\mit i\rm=1}\mit^{K}\beta_i)(\rm\Psi \beta\mit_k\gamma\rm_0+\sum_{\mit i\rm=1}\mit^{K}\beta_i)},
\end{equation}
\end{small}where {\small$\Upsilon=(1-\rho)(N\beta_k+\widetilde{\zeta}\sum_{j\neq k}\beta_j)\xi+\rho\sum_{i=1}^{K}\beta_i$}, {\small$\Psi=\rho \sum_{i=1}^{K}\beta_i+(1-\rho)\xi\widetilde{\zeta}\sum_{j\neq k}\beta_j$}, and {\small$\widetilde{\zeta}=\frac{1+\zeta^2}{1-\zeta^2}$}. 
The expression of {\small$\frac{\partial \overline{C}}{\partial \rho}$} is shown in (24), on the top of the next page.
\newcounter{mytempeqncnt}
\begin{figure*}[!t]
\normalsize
\setcounter{mytempeqncnt}{\value{equation}}
\setcounter{equation}{23}
\begin{small}
\begin{equation}
\label{eqn_dbl_x}
\frac{\partial \overline{C}}{\partial \rho}=-\frac{M\xi\phi \beta_k\widetilde{\zeta}[(1-a \widetilde{\zeta})\phi\kappa^2+\varpi \widetilde{\zeta}]/\sum_{i=1}^{K}\beta_i}
{\rm ln2(1-\rho)^2\mit\big[(a \widetilde{\zeta}-1)\phi\kappa\rm^2+\varpi \widetilde{\zeta}\big]
\big\{\big[(\mit a\kappa\rm^2-\mit M\xi\kappa\beta_k/\sum_{i\rm=1}\mit^{K}\beta_i)\widetilde{\zeta}-\kappa\rm^2\big]\phi+\varpi \widetilde{\zeta}\big\}}
\end{equation}
\end{small}
\setcounter{equation}{\value{mytempeqncnt}}
\hrulefill
\end{figure*}
\setcounter{equation}{24}Assuming $ab\ll1$ for typical massive MIMO systems, (24) can be simplified as
\begin{small}
\begin{equation}
\frac{\partial \overline{C}}{\partial \rho}=-\frac{M\xi\phi \beta_k\widetilde{\zeta}/\sum_{i=1}^{K}\beta_i}
{\rm ln2(1-\rho)^2
\big\{\big[\kappa-(\mit M\xi\beta_k/\sum_{i\rm=1}\mit^{K}\beta_i-a\kappa)\widetilde{\zeta}\big]\kappa\phi\big\}}.
\end{equation}
\end{small}Focusing on the impact of $\zeta$, we observe that {\small$\frac{\partial \overline{C}}{\partial \rho}<0$} and decreases with $\zeta$, while {\small$\frac{\partial \underline{R}_k}{\partial \rho}<0$} and increases with $\zeta$. Therefore, {\small$\frac{\partial \underline{R}_{\rm sec}}{\partial \rho}$} is an increasing function of $\zeta$, which completes the proof.
\end{proof}

\begin{remark} 
From (23) and (25), it shows that {\small$\frac{\partial \underline{R}_{\rm sec}}{\partial \rho}<0$} and {\small$\frac{\partial \underline{R}_{\rm sec}}{\partial \rho}$} is a monotonically increasing function in terms of the level of spatial correlation $\zeta$. It implies that the eavesdropper’s capacity $\overline{C}$ degrades faster than $\underline{R}_k$ does at large $\zeta$. Thus, we conclude that there exists a threshold of correlation coefficient, i.e., $\overline{\zeta}$, where lower-resolution DACs achieve a higher secrecy rate for $\zeta\in(\overline{\zeta},1)$. The value of $\overline{\zeta}$ is obtained from the solution of {\small$\frac{\partial \underline{R}_{\rm sec}}{\partial \rho}=0$} by focusing on the impact of spatial correlation. Note that the higher the correlation the lower the effective dimension (d.o.f.), in the extreme case of $\zeta=1$, the users and Eve are separated only in the angle of arrival domain, which only has dimension 1 instead of $N$. Therefore, quantization noise from lower-resolution DACs could compensate for AN to improve secrecy rate under spatially correlated channel.
\end{remark}

\section{Numerical Results}
In this section, the analytical results are validated through Monte-Carlo simulation. We consider a system with $N=256$, $K=16$, and $M=4$ in all simulations. The large-scale fading is modeled as $\beta_k=(d\rm_{ref}/\mit d_k)^\eta$, where $\eta=3.8$ denotes the path loss exponent, $d\rm_{ref}=300$ (m) and $d_k\leq500$ (m) are, respectively, the reference distance and the distance between the BS and the $k$th user. The expected values in (14) were evaluated by averaging over 1000 random channel realizations.

\begin{figure}[H]
\centering
\centerline{\includegraphics[width=0.26\textwidth]{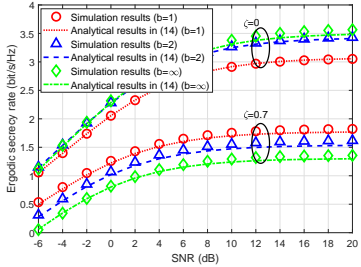}}
\caption{Ergodic secrecy rate and analytical lower bound versus SNR for different spatial correlation coefficient $\zeta$ ($\xi=0.7$)}
\end{figure}

Fig. 1 shows the ergodic secrecy rate versus the average SNR $\gamma_0$ under different DAC resolutions and spatial correlations. The derived lower bound on the secrecy rate is fairly accurate and tight for the entire range of SNR. In addition, it is observed that the secrecy rate is decreasing as $\zeta$ increases. 

Fig. 2 plots the ergodic secrecy rate as a function of the power allocation factor $\xi$. The optimal power allocation factor $\xi^{*}$ largely depends on $\zeta$. Specifically, It is observed that $\xi^{*}$ decreases with $\zeta$. The information signal leakage grows when the spatial correlation is strong. Thus, more power should be allocated for AN to ensure secure communication.

In Fig. 3 (a) we show the ergodic secrecy rate versus $\zeta$ with different DAC resolutions for $\gamma_0$ = 10 dB. We choose a fixed power allocation factor $\xi$ due to the difficulties in optimizing $\xi$ theoretically. The secrecy rate loss due to low-resolution DACs decreases with $\zeta$ as predicted in \emph{Remark 2}. Interestingly, although the channel correlation has a detrimental effect on the secrecy rate, the use of 1-bit DACs can improve the secrecy rate when the spatial correlation coefficient is large. This is because the additional quantization noise serves to increase the level of AN, which is beneficial for spatially colored channels if the AN level has not already been optimized. Finally, Fig. 3 (b) presents the secrecy rate versus $\zeta$ assuming the optimal power allocation $\xi^{*}$ in (17) is chosen. The secrecy rate gaps are $\Delta R_{\rm sec}=0.697$ bit/s/Hz at $\zeta=0$ and $\Delta R_{\rm sec}=0.434$ bit/s/Hz at $\zeta=0.8$, respectively. If optimal power allocation is adopted, then using infinite-resolution DACs can always achieve a higher secrecy rate. In this case, quantization noise from lower-resolution DACs does not compensate for the AN anymore. However, we observe that the secrecy rate loss due to low-resolution DACs decreases with channel correlation coefficient, regardless of the value of $\xi$.

For comparison, Fig. 4 plots the Monte-Carlo simulation by using the spatial correlation model in [9], denoted by $[{\bf R}]_{s,m}=\frac{\beta}{L}\sum_{l=1}^{L}e^{j\pi(s-m)\sin(\varphi_l)}e^{-\frac{\Delta^2}{2}\big(\pi(s-m)\cos(\varphi_l)\big)^2}$, where $\beta$ is the large scale fading coefficient, $\varphi$ is the actual angle-of-arrival and $\Delta$ is the azimuth angular spread. We consider $L=10$ scattering clusters and $\varphi\sim[\frac{-\Delta}{2}$, $\frac{\Delta}{2}]$. It is observed that transitioning from larger to smaller angular spread ($\Delta={\rm50^o}$ to $\Delta={\rm12^o}$) significantly reduces the secrecy rate of the $k$th user for different DAC resolutions. However, the lower resolution DAC is always beneficial for secrecy rate with a fixed $\xi$ under highly correlated channels as expected.

\begin{figure}[H]
\centering
\centerline{\includegraphics[width=0.26\textwidth]{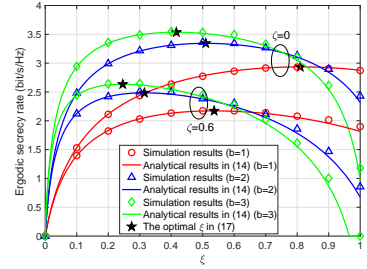}}
\caption{Achievable ergodic secrecy rate versus the power allocation factor $\xi$ for different DAC resolutions ($\gamma_0=10$ dB)}
\end{figure}

\begin{figure}[H]
  \subfigure[with fixed $\xi$=0.7]{
    \label{fig:subfig:a} 
    \includegraphics[width=1.67in]{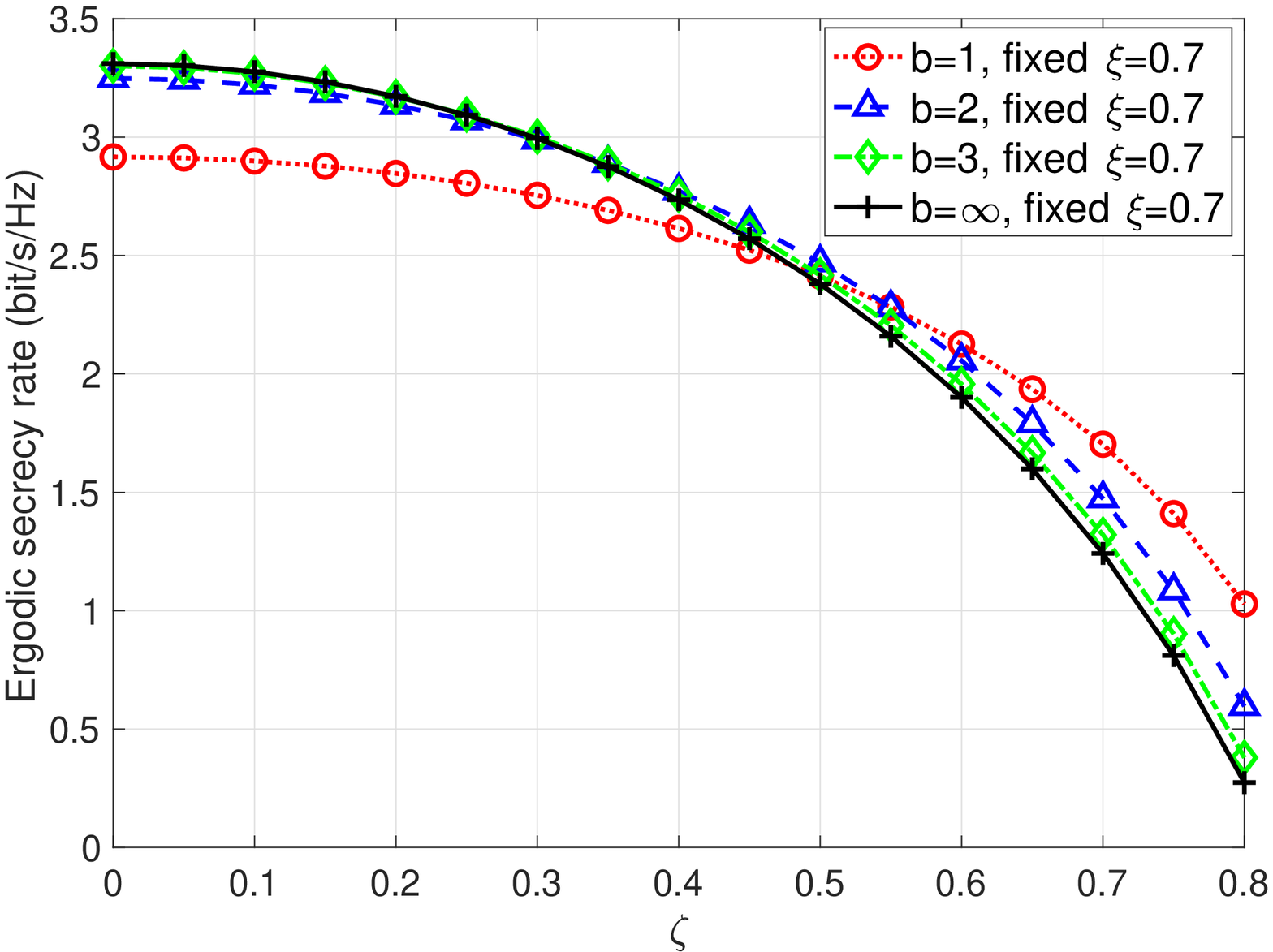}}
  \subfigure[with optimal $\xi^{*}$ in (17)]{
    \label{fig:subfig:b} 
    \includegraphics[width=1.67in]{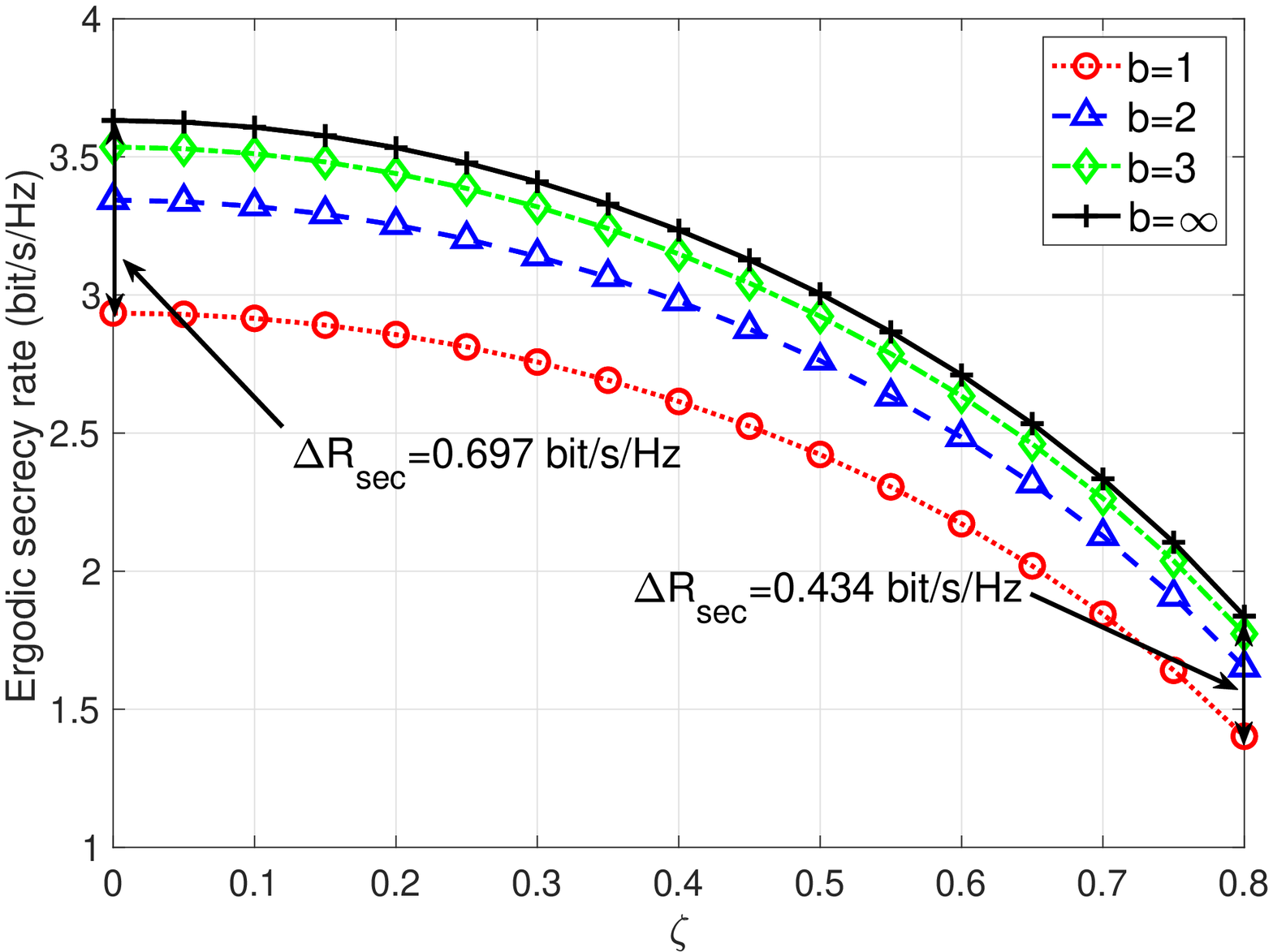}}
  \caption{Achievable ergodic secrecy rate versus $\zeta$ for different DAC resolutions}
  \label{fig:subfig} 
\end{figure}

\begin{figure}[H]
\centerline{\includegraphics[width=0.26\textwidth]{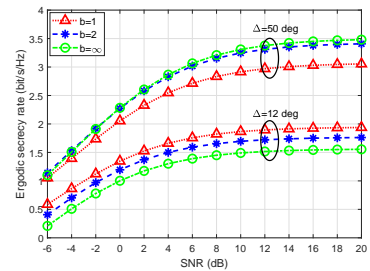}}
\caption{Ergodic secrecy rate versus SNR ($\xi=0.7$)}
\end{figure}

\section{Conclusion}
This paper has characterized the performance of AN-based secure transmission in a massive MIMO downlink system with low-resolution DACs under spatially correlated channels. In particular, it is shown that optimal secrecy performance can be obtained by increasing the amount of power dedicated to artificial noise when the channel correlation increases. Furthermore, the use of low-resolution DACs has been shown to be beneficial to the secrecy performance for a fixed power allocation factor when the channels possess strong spatial correlation. Interesting future extension of this paper includes studying the impact of different spatial correlation matrices at both transmitter and the eavesdropper.

\vspace{-0.3cm}
\section*{Appendix A}
Consider MF precoding satisfying {\small$\rm tr(\bf{W}\bf{W}\mit^H)=K$}, which leads to
{\small$\bf{W}=\mit\sqrt{\frac{K}{N\sum_{i=\rm1}\mit^{K}\beta_i}}\bf{H}$}. First, we directly obtain 
\begin{small} 
\begin{equation}
\begin{aligned}
\begin{split}
\big|\bf{h}\mit_k^T\bf{w}\mit_k\big|\rm^2&=\mit\frac{K}{N\sum_{i=\rm1}\mit^{K}\beta_i}\big|\bf{h}\mit_k^T\bf{h}\mit_k\big|\rm^2\\&=\mit\frac{K\beta\rm^2\mit_k}{N\sum_{i=\rm1}\mit^{K}\beta_i}\big[\rm tr(\bf{R})\big]\rm^2\xrightarrow{a.s.}\mit\frac{KN\beta\rm^2\mit_k}{\sum_{i=\rm1}\mit^{K}\beta_i},
\end{split}
\end{aligned}
\end{equation}
\end{small}where we have used {\small$\frac{1}{\sqrt{N}}\widetilde{\bf{h}}\mit_k^T\bf{R}\rm\frac{1}{\sqrt{\mit N}}\widetilde{\bf{h}}\mit_k^*-\frac{\rm1}{\mit N}\rm tr(\bf{R})\rm\xrightarrow{a.s.}0$} in [16, Lemma 4].
Then, the inter-user interference is calculated as
\begin{small} 
\begin{equation}
\begin{aligned}
\begin{split}
\varrho&=(1-\rho)\mu\sum\limits_{j\neq k}\mit\frac{K}{N\sum_{i=\rm1}\mit^{K}\beta_i}\big|\bf{h}\mit_k^T\bf{h}\mit_j\big|\rm^2
\\&\xrightarrow{a.s.}(1-\rho)\mu K\beta_k\rm tr(\bf{R}\rm^2)\mit\sum_{j\neq k}\beta_j\bigg/\bigg(N\sum_{i=\rm1}\mit^{K}\beta_i\bigg).
\end{split}
\end{aligned}
\end{equation}
\end{small}

For large $N$ and $K$, {\small$\bf{C}_q$} converges to
\begin{small}
\begin{equation}
\bf{C}_q\mit\xrightarrow{a.s.}\rho\frac{P}{N}\bf{I}\mit_N,
\end{equation}
\end{small}where we use the definition of $\mu$ and $\nu$, and the fact that {\small$\rm diag(\bf{W}\bf{W}\mit^H)\xrightarrow{a.s.}\frac{K}{N}\bf{I}\mit_N$} and {\small$\rm diag(\bf{V}\bf{V}\mit^H)\xrightarrow{a.s.}\frac{N-K}{N}\bf{I}\mit_N$} due to the strong law of large numbers. 
Further, we obtain the component of the quantization noise as
\begin{small}
\begin{equation}
\bf{h}\mit_k^T\bf{C_q}\bf{h}\mit_k^*\xrightarrow{a.s.}\rho\frac{P}{N}\beta_k \rm tr(\bf{R})=\mit\rho P \beta_k.
\end{equation}
\end{small}Regarding the AN power, it follows that
\begin{small}
\begin{equation}
\bf{h}\mit_k^T\bf{V}\bf{V}\mit^H\bf{h}\mit^*_k=\rm0,
\end{equation}
\end{small}since {\small$\bf{H}\bf{V}=0$}. Finally, by substituting (26), (27), (29), (30) and the definition of $\mu$ and $\nu$ into (8), and according to the Continuous Mapping Theorem, we complete the proof.

\vspace{-0.3cm}
\section*{Appendix B}
By applying Jensen's inequality, the capacity of the eavesdropper can be upper bounded as 
\begin{equation}
C\le\rm log_2\big[1+(1-\rho)\mu\mathbb{E}\big\{\bf{w}\mit_k^H\bf{H}\rm_e\mit^H\bf{X}\rm^{-1}\bf{H}\rm_e\bf{w}\mit_k\big\}\big].
\end{equation}
Let us first focus on the term {\small$\bf{X}$} and by substituting (28) into (12) yields
\begin{small}
\begin{equation}
\bf{X}\mit\xrightarrow{a.s.}\bigg[(\rm1-\mit\rho)\nu+\rho\frac{P}{N}\bigg]\bf{X}\rm_1+\mit\rho\frac{P}{N}\bf{X}\rm_2,
\end{equation}
\end{small}where {\small$\bf{X}\rm_1=\bf{H}\rm_e\bf{V}\bf{V}\mit^H\bf{H}\rm_e\mit^H$} and {\small$\bf{X}\rm_2=\bf{H}\rm_e\bf{V}\rm_0\bf{V}\rm_0\mit^H\bf{H}\rm_e\mit^H$}. It is obvious that {\small$[\bf{V}\ \bf{V}\rm_0][\bf{V}\ \bf{V}\rm_0]\mit^H=\bf{I}\mit_M$}, because {\small$[\bf{V}\ \bf{V}\rm_0]$} forms a complete orthogonal basis. 
Eigendecompose {\small$\bf{R}$} such that {\small$\bf{R}=\bf{U}\bf{\Lambda}\bf{U}\mit^H$} to decorrelate matrix {\small$\bf{H}\rm_e$} as {\small$\bf{Z}=\bf{H}\rm_e\bf{\Lambda}\rm^{-\frac{1}{2}}\bf{U}\mit^H$}, where {\small$\bf{\Lambda}=\rm \widetilde{diag}(\lambda\rm_1\mit,...,\lambda_N)$} is the diagonal matrix of the eigenvalues of {\small$\bf{R}$} and the columns of {\small$\bf{U}$} consist of the corresponding eigenvectors. Since {\small$\bf{U}$} is unitary, the statistics of {\small$\bf{Z}\bf{U}$} are identical to those of {\small$\bf{Z}$}. Thereby, the distributions of {\small$\bf{X}\rm_1$} and {\small$\bf{X}\rm_2$} are the same as
\begin{small}
\begin{equation} 
\mit\sum_{i=\rm1}\mit^{N}\mit\sum_{j=\rm1}\mit^{N}\lambda_i\rm^{\frac{1}{2}}\lambda\mit_j\rm^{\frac{1}{2}}\bf{z}\mit_i\bf{v}\mit_{i}\bf{v}\mit_{j}^H\bf{z}\mit_j^H
\end{equation}
\end{small}and
\begin{small}
\begin{equation} 
\mit\sum_{i=\rm1}\mit^{N}\mit\sum_{j=\rm1}\mit^{N}\lambda_i\rm^{\frac{1}{2}}\lambda\mit_j\rm^{\frac{1}{2}}\bf{z}\mit_i\bf{v}\rm_{0,\mit i}\bf{v}\rm_{0,\mit j}\mit^H\bf{z}\mit_j^H,
\end{equation}
\end{small}where {\small$\bf{z}\mit_i$} is the $i$th row of {\small$\bf{Z}$, $\bf{v}\mit_i$} and {\small$\bf{v}\rm_{0,\mit i}$} are $i$th column of {\small$\bf{V}$} and {\small$\bf{V}\rm_0$}, respectively. Following the same approach in [17], {\small$\bf{Y}=\rm\big[(1-\mit \rho)\nu+\rho\frac{P}{N}\big]\bf{Y}\rm_1+\mit\rho\frac{P}{N}\bf{Y}\rm_2$} may be accurately approximated as a single scaled Wishart matrix {\small$\bf{Y}\mit\sim\mathcal{W}_M(\eta,\varphi\bf{I}\mit_M)$}, where we define 
${\small\bf{Y}\rm_1}=\mit\sum_{m=\rm1}\mit^{N}\lambda_m\bf{z}\mit_m\bf{v}\mit_{m}\bf{v}\mit_m^H\bf{z}\mit_m^H$ and ${\small\bf{Y}\rm_2}=\mit\sum_{n=\rm1}\mit^{N}\lambda_n\bf{z}\mit_n\bf{v}\rm_{0,\mit n}\bf{v}\rm_{0,\mit n}\mit^H\bf{z}\mit_n^H$. Equating the first two moments of those matrices with {\small$\bf{Y}\rm_1\sim\mit\sum_{m=\rm1}\mit^{N}\lambda_m\mathcal{W}_M(N-K,\rm\frac{1}{\mit N}\bf{I}\mit_M)$} and {\small$\bf{Y}\rm_2\sim\mit\sum_{n=\rm1}\mit^{N}\lambda_n\mathcal{W}_M(K,\rm\frac{1}{\mit N}\bf{I}\mit_M)$} leads to
\begin{small} 
\begin{equation}
\mit\eta\varphi=\bigg[\rm(1-\mit\rho)\nu+\rho\frac{P}{N}\bigg](N-K)+\rho\frac{P}{N}K
\end{equation}
\end{small}and
\begin{small} 
\begin{equation}
\mit\eta\varphi\rm^2=\frac{\rm tr(\bf{R}\rm^2)}{\mit N}\bigg\{\bigg[(\rm1-\mit\rho)\nu+\rho\frac{P}{N}\bigg]\rm^2\mit(N-K)+\bigg(\rho\frac{P}{N}\bigg)\rm^2\mit K\bigg\},
\end{equation}
\end{small}where we use {\small$\sum_{i=1}^{N}\lambda_i=\rm tr(\bf{R})$} and {\small$\sum_{i=1}^{N}\lambda_i^2=\rm tr(\bf{R}\rm^2)$}. By exploiting the independence of the elements in {\small$\widetilde{\bf{H}}\rm_e$}, we can further obtain {\small$\bf{X}\rm^{-1}\xrightarrow{a.s.}1/(\varphi(\mit \eta-M))\bf{I}\mit_M$} with {\small$\eta>M$}, where we use the property {\small$\bf{A}\rm^{-1}\xrightarrow{a.s.}1/\mit(n-m)\bf{I}\mit_m$} for a Wishart matrix {\small$\bf{A}\sim\mathcal{W}\mit_m(n,\bf{I}\mit_m)$} with {\small$n>m$} [4]. Substituting this result and {\small$\mathbb{E}\big[\bf{w}\mit_k^H\bf{H}\rm_e\mit^H\bf{H}\rm_e\bf{w}\mit_k\big]=\frac{MK\beta_k}{N\sum_{i=\rm1}\mit^{K}\beta_i}\rm tr(\bf{R}\rm^2)$} into (31) completes the proof.

\end{document}